\documentclass[%
 prb,reprint,10pt,
 amsmath,amssymb,
 aps,
]{revtex4-2}

\usepackage[dvipdfmx]{graphicx}
\usepackage{dcolumn}
\usepackage{bm}
\usepackage{braket}
\usepackage[dvipsnames]{xcolor}
\usepackage{amsthm}
\theoremstyle{definition}
\usepackage[percent]{overpic}
\usepackage[version=4]{mhchem}
\usepackage{amsmath,amssymb}
\usepackage{tikz}
\usepackage{booktabs}

\usepackage{bm}

\usepackage{array}      
\usepackage{seqsplit}   

\newcolumntype{L}[1]{>{\raggedright\arraybackslash}p{#1}}
\newcolumntype{T}[1]{>{\ttfamily\raggedright\arraybackslash}p{#1}}

\usetikzlibrary{arrows.meta,patterns,calc}

\begin{document}

\preprint{APS/123-QED}

\title{
Density functional theory for core-level X-ray absorption
}

\author{Seokkyu An}

\author{Taisuke Ozaki}%

\affiliation{%
 Institute for Solid State Physics, The University of Tokyo, Kashiwa 277-8581, Japan
}%

\date{\today}

\begin{abstract}
We establish a rigorous density functional theory (DFT) framework for core-level X-ray absorption spectroscopy (XAS) by formulating a constrained search for core-excited states based on the Gunnarsson-Lundqvist theorem.
Within this framework, the explicit-core $\Delta$SCF scheme enables shift-free absolute edge alignment and a consistent treatment of L/M edges with spin-orbit-resolved projectors.
In addition, by exploiting dipole selection rules, we recast the evaluation of the dipole matrix elements, which otherwise requires many independent Slater determinant calculations, into a compact single determinant form.
This reduces the computational scaling from $\mathcal{O}(N^4)$ to $\mathcal{O}(N^3)$, where $N$ is the number of electrons, without introducing additional approximations.
Across representative \ce{C}, \ce{B}, \ce{O}, and \ce{Li} K-edge benchmarks in molecules and solids, the method reproduces line shapes, polarization anisotropies, and absolute onsets without empirical shifts, providing a robust and scalable route to quantitatively reliable XAS simulations within DFT.
\end{abstract}

\maketitle

\newcommand{\ncore}{\hat{n}_\mathbf{c}}
\newcommand{\Hamiltonian}{\hat{H}}
\newcommand{\corestate}{\ket{\phi_\mathbf{c}}}
\newcommand{\polarizedcorestate}{\ket{\hat{O}\phi_\mathbf{c}}}
\newcommand{\createcore}{\hat{a}^\dagger_\mathbf{c}}

\newcommand{\Finalstate}{\Psi_\mathrm{f}}
\newcommand{\Initialstate}{\Psi_\mathrm{i}}
\newcommand{\FinalKSstate}{\Phi_\mathrm{f}}
\newcommand{\InitialKSstate}{\Phi_\mathrm{i}}
\newcommand{\finalstate}[1]{\psi_\mathrm{f#1}}
\newcommand{\initialstate}[1]{\psi_\mathrm{i#1}}
\newcommand{\dipolematrix}{\braket{\Finalstate|\hat{O}|\Initialstate}}
\newcommand{\Newmatrixelement}{\braket{\FinalKSstate,\phi_\mathbf{c}|\InitialKSstate,\hat{O}\phi_\mathbf{c}}}

\section{Introduction}
X-ray absorption spectroscopy (XAS) is a standard probe of local electronic structure owing to its element and orbital specificity, compatibility with \textit{in situ} measurements, and sensitivity to bonding and symmetry \cite{Henderson2014_RMG_XANES,Severino2024NatPhoton, Lindblad2022_PRA_NOplus_NEXAFS, Fan2025_AdvSci_OperandoXAS, Frati2020_ChemRev_O_Kedge}. While band structures and projected densities of states (PDOS), which are common diagnostics, can indicate trends, the reproduction of line shapes, absolute edge positions, polarization dependence, and multiplet features requires a first-principles framework that treats many-body relaxation and enforces dipole selection rules on equal footing \cite{Timrov2020_LaFeO3_Ni_PRResearch,Preston2007DyN_SmN_PRB,Kang2005_RMnO3_PES_XAS_PRB}. Consequently, \textit{ab initio} simulations are indispensable both for the quantitative prediction of XAS, enabling direct comparison to experimental data, and for the microscopic interpretation of many-body relaxation and selection-rule-driven features.

Among the available strategies, $\Delta$SCF methods based on  density functional theory (DFT) \cite{Hohenberg1964PR} are now widely used as a robust and cost-effective approach \cite{Liang2017PRL,Liang2018PRB,Roychoudhury2021_ACSAMI_LiK,Meng2024_PRM_TiK_Benchmark,Woicik2020_PRB_CoreHole,Roychoudhury2023PRB,JanaHerbert2023_JCTC_TP, Roychoudhury2022_PRB_XES_Polarization}.
By explicitly constructing a core-excited state and relaxing the electron density self-consistently, the $\Delta$SCF method captures nonperturbative orbital relaxation beyond response-based treatments (e.g., Bethe-Salpeter equation (BSE) \cite{Vinson2011_BSE_CoreXAS_PRB,Vinson2012_TM_Ledge_BSE_PRB} or linear-response time-dependent DFT (LR-TDDFT) \cite{Besley2010_TDDFT_CoreSpectroscopy_PCCP,Andrade2015_Octopus_RealSpace_PCCP}). 
In practice, using single Kohn-Sham (KS) Slater determinants to evaluate dipole matrix elements often yields improved edge positions and line shapes at modest cost \cite{Liang2017PRL,Liang2018PRB,Roychoudhury2021_ACSAMI_LiK,Meng2024_PRM_TiK_Benchmark}.
However, the lack of a clear definition of what qualifies as `core-excited DFT' has limited its generality and extensibility.
Many prior implementations introduce the core-hole pseudopotential (PP) with positive charge \cite{Liang2017PRL, Liang2018PRB, Roychoudhury2021_ACSAMI_LiK, Roychoudhury2023PRB}, resulting in the enforcement of a spherically symmetric core channel. While this suffices at K edges, it is problematic for $j$-resolved L/M edges. 
More importantly, when different PPs are used for the initial and final states, the two calculations no longer share a common energy reference; consequently, absolute edge positions typically must be corrected by ad hoc rigid shifts, which hinders comparisons across different compounds.
In addition, despite recent progress \cite{Liang2018PRB,Roychoudhury2023PRB}, the evaluation of the dipole matrix element between Slater determinants remains a non-negligible bottleneck for large supercells. 
Thus, a useful core-excited DFT framework should
(i) define a core-excited state in a variationally controlled
manner while preserving a common energy reference, and
(ii) provide an intensity formula whose computational
scaling and structure are suitable for large-scale simulations. 
To address these issues, in this paper we present two
developments, summarized below.

\textbf{(i) Explicit-core $\Delta$SCF via a penalty projector.}
We establish a formally rigorous DFT framework for the lowest core-excited state by formulating a constrained search based on the Gunnarsson-Lundqvist theorem \cite{GL}.
Then, the combination with DFT yields the penalized KS Hamiltonian that reproduces the lowest core-excited energy,
\begin{equation}
  \hat h_{\Delta} \;=\; \hat h_\mathrm{KS} \;+\; \Delta\, \ncore, 
  \qquad \Delta \gg 1,
  \label{eq:penalty-intro}
\end{equation}
where $\hat{h}_\text{KS}$ is the conventional KS Hamiltonian and $\ncore$ is the number operator for a chosen atomic core orbital $\corestate$ (see Sec. \ref{sec:Delta-SCF}). 
The second penalty term has appeared in several constrained(c)-DFT studies, even if not always stated in this formal way \cite{Derricotte2015_PCCP_OCDFT_XAS,Ozaki2017PRL}. 
Then, by treating the core explicitly rather than using core-hole PPs, we preserve absolute energy alignment, and seamlessly generalize to non-collinear magnetism and to L/M edges by choosing a spinor $j$-resolved $\corestate$. 

\textbf{(ii) Selection-rule–aware reformulation of absorption intensities.}
In the most recent work, the dipole matrix element was reformulated by the following \cite{Roychoudhury2023PRB}:
\begin{equation}
\dipolematrix= \sum_{n\in \text{occ}}  o_{n,\mathbf{c}}^{(\text{CHB})}\braket{\Phi_{\mathrm{f}}^{(n\to \mathbf{c})}| \InitialKSstate },
\end{equation}
so-called core-hole basis (CHB) treatment (see Sec. \ref{sec:reformulation}).
While it overcomes the shortcoming of the sudden approximation and reduces $\sim$50~\% computational time compared to previous ground-state basis (GSB) \cite{Liang2017PRL,Liang2018PRB}, its time complexity scales as $\mathcal{O}(N^4)$ for large systems, where $N$ is the number of electrons.
To further improve scalability, we reformulate the matrix element into a more compact and symmetry-transparent form that eliminates explicit sums over many intermediate Slater determinants,
\begin{equation}\label{eq:pre_new_formalism}
\dipolematrix=-\Newmatrixelement,
\end{equation}
where $\ket{\Phi_{\mathrm f},\phi_\mathbf{c}}$ and $\ket{\Phi_{\mathrm i},\hat{O}\phi_\mathbf{c}}$ are $(N{+}1)$-electron Slater determinants obtained by adding $\corestate$ and $\ket{\hat{O}\phi_\mathbf{c}}$ (the ``polarized core'' state). Compared to previous treatments \cite{Liang2017PRL,Liang2018PRB,Roychoudhury2023PRB}, Eq. \eqref{eq:pre_new_formalism} yields a simpler computational structure with $\mathcal{O}(N^{3})$, enabling efficient calculations for large systems. Additionally, by exploiting linear properties of the determinant, we eliminate additional iterations over many final-state configurations, allowing a single determinant evaluation per each targeted state.

The structure of this paper is as follows: In Sec. II, we present the core-excited DFT and reformulation of the matrix elements. Sec. III provides a detailed discussion on the implementation of the method. A series of benchmark calculations are presented in Sec. IV. Finally, in Sec. V, we summarize our study.

\section{Theory}

\subsection{$\Delta$SCF with an explicit core projector}\label{sec:Delta-SCF}
When a high-frequency X-ray photon impinges on a material, a deeply bound core electron can absorb the photon and be promoted into an unoccupied state above the Fermi level. X-ray absorption spectroscopy (XAS) measures the corresponding absorption intensity $\sigma(\omega)$, given by Fermi’s golden rule \cite{deGrootKotani2008},
\begin{equation}\label{eq:Fermi_golden_rule}
  \sigma(\omega)\;\propto\sum_\text{f}|\dipolematrix|^2\,
  \delta\big(\hbar\omega-(E_\text{f}-E_\text{i})\big),
\end{equation}
where $\ket\Initialstate$ and $\ket\Finalstate$ are many-body initial (ground) and final (core-excited) eigenstates and 
$\hat O = \hat{\boldsymbol\epsilon}\!\cdot\!\hat{\bm{p}}$ is the dipole operator with polarization $\hat{\boldsymbol\epsilon}\in\{\hat{\mathrm{x}},\hat{\mathrm{y}},\hat{\mathrm{z}}\}$.
The relevant final states with the core hole are characterized by $\braket{\ncore} = 0$, whereas the initial state $\braket{\ncore} = 1$, where $\ncore$ is the number operator of the atomic core state $\corestate$. Inspired by this observation, we assume that $\ncore$ commutes with the many-body Hamiltonian, 
\begin{equation}
  [\,\hat H,\;\ncore\,]\simeq0,
  \label{eq:comm-assump}
\end{equation}
which can be justified when the interatomic hybridization of the core state is negligible because of its spatially and energetically well-localized nature.
Then, the following Theorem 1 shows that the variational principle in the sub-Hilbert space where $\braket{\ncore}=0$ gives the lowest core-excited energy, and Theorem 2 shows that the indirect minimization provides the same result. 

\newtheorem{theorem}{Theorem}

\begin{theorem}[Gunnarsson--Lundqvist (GL) \cite{GL}]
Let $\hat{H}$ be the many-body Hamiltonian and $\hat{{A}}$ an operator such that 
$[\hat{H},\hat{{A}}]=0$. 
For a fixed eigenvalue $\lambda$ of $\hat{{A}}$, let 
$\mathbb{H}_\lambda:=\{\,|\Psi\rangle:\hat{{A}}|\Psi\rangle=\lambda|\Psi\rangle\,\}$ 
denote the corresponding subspace and let 
$E^{(\lambda)}_0:=\min_{|\Psi\rangle\in\mathbb{H}_\lambda}\langle\Psi|\hat{H}|\Psi\rangle$.
Then, $E^{(\lambda)}_0$ coincides with the lowest excited energy, and the corresponding eigenstate satisfies $\braket{\hat{{A}}}=\lambda$. 
\label{theorem1}
\end{theorem}

\begin{theorem}[Penalty functional]\label{thm:penalty}
Let $\hat{P}$ be the projector onto $\mathbb{H}_\lambda$ and $\hat{Q}:=\hat{\mathbf{1}}-\hat{P}$.
For $\Delta>0$, define the penalized Hamiltonian 
$\hat{H}_\Delta:=\hat{H}+\Delta\, \hat{Q}$.
Then, for sufficiently large $\Delta$,
the ground state energy of $\hat{H}_\Delta$ coincides with the lowest excited energy, and the corresponding eigenstate satisfies $\braket{\hat{P}}=1$.
\label{theorem2}
\end{theorem}

\begin{proof}[Proof of Theorem \ref{theorem2}]
\mbox{}\\[0.2em]
By the assumption $[\hat{H},\hat{{A}}]=0$, there exists a basis set $\{\ket{\Psi_{i,P}}\}\cup\{\ket{\Psi_{j,Q}}\}$, which satisfies $\Hamiltonian\ket{\Psi_{i,P}} =E_i\ket{\Psi_{i,P}}$ and $\hat{P}\ket{\Psi_{i,P}} =\ket{\Psi_{i,P}}$ ($\ket{\Psi_{j,Q}}$ defined similarly). Then, for any state $\ket\Psi=\sum_i{c_i\ket{\Psi_{i,P}}}+\sum_j{d_j\ket{\Psi_{j,Q}}}$ with $\sum_i|c_i|^2+\sum_j|d_j|^2=1$, the following variational principle holds:
\begin{align}\label{eq:theorem2_proof}
E_\mathrm{f}
&:= \min_{\ket\Psi}\braket{\Psi|\hat{H}_\Delta|\Psi}\nonumber\\
&=\min_{\{c_i,d_j\}}\Big(\sum_i|c_i|^2 E_i+\sum_j|d_j|^2 (E_j+\Delta)\Big)\nonumber \\
&=\min\!\left\{\min_i E_i,\;\min_j(E_j+\Delta)\right\}\nonumber\\
&=\min_i E_i \quad\text{if }\min_i E_i < \min_j (E_j + \Delta) .
\end{align}
For second line, the variational principle is applied to the two cases that the minimum energy appears in $E_i$ and $E_j+\Delta$, respectively. 
This proves Theorem 2 and Theorem 1 by setting $d_j=0$ for all $j$.
\end{proof}
\noindent Consequently, by choosing $\hat{Q}=\ncore$, the ground state of the penalized many-body Hamiltonian,
\begin{equation}\label{eq:penalized many-body Hamiltonian}
\Hamiltonian_\Delta :=\Hamiltonian +\Delta\ncore
\end{equation}
corresponds to the lowest core-excited state which satisfies $\braket{\ncore}=0$.
Next, Theorem \ref{thm:levy} shows that Eq. \eqref{eq:penalized many-body Hamiltonian} is compatible with DFT \cite{Hohenberg1964PR}.

\begin{theorem}[Density functional (Levy's constrained search \cite{Levy1982_PRA_ConstrainedSearch})]\label{thm:levy}
For $\hat{H}_\Delta=\hat T+\hat V_\mathrm{ee}+\hat V_{\rm ext}+\Delta \hat Q$ with the fixed projector $\hat Q$,
the lowest core-excited energy is the minimum of the density functional, 
\begin{align}
E_\mathrm{f}
&=\min_{n}\Big\{\,F_\Delta[n]+\int v_{\rm ext}(\mathbf r)\,n(\mathbf r)\,d\mathbf r\,\Big\},\\
F_\Delta[n]&:=\min_{\Psi\to n}\langle\Psi|\,\hat T+\hat V_\mathrm{ee}+\Delta \hat Q\,|\Psi\rangle,
\end{align}
where $F_\Delta[n]$ is universal (independent of $v_{\rm ext}$) for a fixed $\hat Q$.
\end{theorem}

\begin{proof}[Proof of Theorem 3]

\begin{align}
E_\mathrm{f}
  &= \min_{|\Psi\rangle}\,\langle \Psi | \hat{H}_\Delta | \Psi \rangle \nonumber\\
  &= \min_{n}\Bigl(\, \min_{|\Psi\rangle \to n}\, \langle \Psi | \hat{H}_\Delta | \Psi \rangle \Bigr) \nonumber\\
  &= \min_{n}\Bigl(\, \min_{|\Psi\rangle \to n}\, 
      \langle \Psi | \hat{T} + \hat{V}_\mathrm{ee} + \Delta \,\hat Q | \Psi \rangle 
      + \int v_{\mathrm{ext}}\, n\, dr \Bigr) \nonumber\\
  &= \min_{n}\Bigl( F_\Delta[n] + \int v_{\mathrm{ext}}\, n\, dr \Bigr). 
\end{align}
Here, ``$\ket{\Psi}\to n$'' indicates that $\ket{\Psi}$ is any state which yields the given density $n$. This proves Theorem 3.
\end{proof}
\noindent
Next, with the KS ansatz of the auxiliary system \cite{KohnSham1965}, the eigenstate of the effective Hamiltonian can be given as a KS Slater determinant $\ket{\Phi}=|\{\psi_n\}|$, and the lowest core-excited state satisfies a local minimum condition with respect to each $\bra{\psi_n}$,
\begin{equation}
   \left[\frac{\delta F_\Delta}{\delta \psi_n^*}  + v_\mathrm{ext}\right]\psi_n=\epsilon_n\psi_n,\label{eq:KS_org}
\end{equation}
where $\epsilon_n$ is a KS eigenenergy as a Lagrange multiplier for the orthonormal condition $\braket{\psi_m|\psi_n}=\delta_{m,n}$. 
As a consequence of Eqs. \eqref{eq:comm-assump}–\eqref{eq:KS_org}, we can explicitly state that, under the assumption of Eq. \eqref{eq:comm-assump}, the first core-excited state is rigorously formulated within DFT. In particular, combining the GL-based variational principle with the KS auxiliary-system ansatz leads to a formally exact $\Delta$SCF description, and yields the core-excited KS equation in Eq. \eqref{eq:KS_org}.
It should be noted that the Theorems 1--3 and the associated KS framework are also valid for ($N-1$)-body core-excited states. 
Thus, the fact establishes the foundation for the previous calculations of absolute binding energies of core levels for the X-ray photoemission spectroscopy (XPS) \cite{Ozaki2017PRL}.

We now decompose $F_\Delta[n[\Phi]]$ into each contribution and assume that the local/semi-local approximation (LDA/GGA) \cite{Perdew1996_PRL_PBE} holds as
\begin{align}
F_{\rm \Delta}[{\Phi}]&=T_s+E_{\rm H}+E_\mathrm{P}+E_{\rm \Delta,xc}, \label{eq:F}\\
E_{\rm \Delta,xc}[{\Phi}]&\simeq E_\mathrm{LDA/GGA},
\end{align}
where $E_\mathrm{P}$ is 
\begin{equation}\label{eq:E_p}
E_\mathrm{P}[{\Phi}]=\braket{\Phi|\Delta\ncore|\Phi}=\sum_{n\in \text{occ}}\braket{\psi_n|\Delta\ncore|\psi_n}.
\end{equation}
Applying Eqs. \eqref{eq:F}--\eqref{eq:E_p} to Eq. \eqref{eq:KS_org}, we have the practical form:
\begin{equation}
{\;\left[\hat h_{\rm KS}[{\Phi^\text{in}}]+\Delta\ncore\right]\psi_n^\text{out}=\epsilon_n\psi_n^\text{out}\;},
\label{eq:ks-penalty}
\end{equation}
where $\hat{h}_\mathrm{KS}$ is the conventional KS Hamiltonian, and $\Delta\ncore=\Delta\ket{\phi_\mathbf{c}}\bra{{\phi_\mathbf{c}}}$ is the penalty term. 
Index ``in/out'' are for the self-consistent-field (SCF) iteration, and the SCF condition is achieved at ${\Phi^\text{out}}={\Phi^\text{in}}=:{\Phi_\mathrm{f}}$, while the penalty term enforces the vacant core occupation $\braket{\ncore}=0$. 
The practical implementation will be discussed in Sec. \ref{section:LCAO}.

So far, we present how to calculate the lowest core-excited energy $E_\mathrm{f}$. 
To obtain the absorption spectra following Eq. \eqref{eq:Fermi_golden_rule}, we should also take into account the higher-lying states, which are $\braket{\ncore}=0$ and orthonormal to the other eigenstates. 
One may approximately find different KS state configurations, formally written as 
\begin{equation}
  \ket{\Phi^{(n\to m)}_\mathrm{f}}:= \hat{a}^\dagger_m\hat{a}_n\ket{\Phi_\mathrm{f}},
\end{equation}
where ``$n\to m$'' denotes the single excitation from the $n$-th to the $m$-th KS state, and further multiple excitations can be defined similarly. 
Note that its orthogonality is ensured by the orthogonality of the KS states, and $\braket{\ncore}=0$ holds unless we include $\corestate$ in the configuration as a final KS state. 
We then evaluate the corresponding energy $E_\mathrm{f}^{(n\to m)}$ non-self-consistently as 
\begin{equation}
    E_\mathrm{f}^{(n\to m)}\simeq E_\mathrm{f}[{\Phi^{(n\to m)}_\mathrm{f}};\;n[{\Phi_\mathrm{f}}]],
\end{equation}
where density and Hamiltonian are no longer updated from those of $\ket{\Phi_\mathrm{f}}$. Although this treatment does not take account of additional relaxations, we expect that it can be negligible compared to the core relaxation, which is already included in the Hamiltonian. 
Then, the energy difference can be evaluated by the KS eigenenergy difference as 
\begin{equation} 
E_\mathrm{f}^{(n\to m)}-E_{\mathrm{f}} \;=\; \epsilon_m - \epsilon_{{n}},\label{eq:Ef_rule} 
\end{equation}
since double counting terms depend only on the fixed density.
The next two subsections will present how to obtain associated dipole matrix elements.

\subsection{Reformulation of the matrix element}\label{sec:reformulation}
As appeared in Eq. (\ref{eq:Fermi_golden_rule}), the dipole matrix element determines the absorption intensity and can be written in second-quantized form as 
\begin{equation}\label{eq:org_mat}
\dipolematrix= \sum_{m,n}o_{m,n}\braket{\Finalstate|\hat{a}_m^\dagger \hat{a}_n|\Initialstate},
\end{equation}
where $o_{m,n}$ is a single-particle dipole matrix element with $m$ and $n$, which run for the chosen orthonormal basis set. 
The final state may be expanded as \cite{Liang2018PRB} 
\begin{align}
\ket{\Finalstate}
&=\sum_{\substack{m\in\text{unocc}\\n\in\text{occ}}} C_{mn}\,\hat a^\dagger_{m}\hat a_{n}\ket{\Initialstate}\nonumber
\\&\quad+\sum_{\substack{m,m'\in\text{unocc}\\n,n'\in\text{occ}}} C_{mm'nn'}\, \hat a^\dagger_{m}\hat a^\dagger_{m'}\hat a_{n'}\hat a_{n}\ket{\Initialstate}\nonumber\\
 &\quad+\,\cdots\,,
\end{align}
where $n$ ($m$) label occupied (unoccupied) single-particle states that diagonalize the reference Hamiltonian such as the initial KS Hamiltonian. 
Thus, response-based formalisms such as  BSE and LR-TDDFT, which describe excitations as fluctuations around the ground-state density, often fail for deep-core edges due to the neglect of substantial higher-order orbital relaxation unless sufficiently rich kernels and bases are employed \cite{Liang2017PRL,Liang2018PRB,Roychoudhury2021_ACSAMI_LiK,Meng2024_PRM_TiK_Benchmark}.

On the other hand, the $\Delta$SCF method approximates both initial and final many-body states to the single KS Slater determinants as
\begin{align}\label{eq:SDapprox}    
\ket{\Psi_{\alpha}}\simeq\ket{\Phi_{\alpha}}=|\{\psi_{\alpha}\}| \quad\quad (\alpha\in\{\mathrm {i,f}\}), 
\end{align}
where $\{{\psi_\alpha}\}$ is the $N$ KS states of the initial (i) and final (f) KS Hamiltonian, respectively. 
Strictly speaking, the KS Slater determinant is {not} the exact many-body eigenstate due to the neglect of multiple configurations, which may lead to further energy splitting.
Nevertheless, since $\InitialKSstate$ ($\FinalKSstate$) reproduces the ground (core-excited) state density and provides a variationally consistent estimate of $E_{\mathrm i}$ ($E_{\mathrm f}$, respectively), the computation of dipole matrix elements in the KS basis systematically incorporates core relaxation and thus surpasses the sudden approximation which omits it. 
Then, it reads 
\begin{equation}\label{eq:GSB}
\braket{ \Psi_{\mathrm{f}} | \hat{O} | \Initialstate }
= \sum_{m\in \text{unocc}} o_{m,\mathbf{c}}^{(\text{GSB})}\braket{\Phi_{\mathrm{f}}| \InitialKSstate^{{(\mathbf{c}\to m)}} },
\end{equation}
where ${o}^\text{(GSB)}_{m,\mathbf{c}}$ is expanded by the initial KS states (ground-state basis, GSB \cite{Liang2017PRL,Liang2018PRB}), and $m$ runs for the unoccupied states. 
Note that $\braket{\Phi_{\mathrm{f}}| \InitialKSstate^{(\mathbf{c} \to {m})} }$ takes account of many-body excitation effects and does not behave as a Kronecker delta function due to the loss of orthonormality between the eigenstates of different Hamiltonian, i.e., $\braket{\psi_{\alpha,m}|\psi_{\alpha,\mathit{n}}}= \delta_{m,n}$ but 
$\braket{\finalstate{,\mathit{m}}|\initialstate{,\mathit{n}}}\neq \delta_{m,n}$. 
Furthermore, Eq. \eqref{eq:GSB} contains the infinite summation over unoccupied states, which requires a large basis set, and thus high computational costs for convergence. 
Thus, the recent work \cite{Roychoudhury2023PRB} has reformulated it by means of the final states expansion, given as
\begin{equation}
\braket{ \Psi_{\mathrm{f}} | \hat{O} | \Initialstate } = \sum_{n\in \text{occ}}  o_{n,\mathbf{c}}^{(\text{CHB})}\braket{\Phi_{\mathrm{f}}^{(n\to \mathbf{c})}| \Phi_\mathrm{i} }, \label{eq:CHB}
\end{equation}
where ${o}^{\text{(CHB)}}_{n,\mathbf{c}}$ is expanded by the final KS states (core-hole basis, CHB \cite{Roychoudhury2023PRB}), and $n$ runs for the occupied states.
This treatment successfully reduces the computational time $\sim$50\% with a smaller basis set, but its time complexity still scales as $\mathcal{O}(N^4)$ for large systems. 
The comparison of each treatment is given in Table. \ref{tab:comp_time}.

\begin{table}
\setlength{\tabcolsep}{3.5pt}
\caption{Comparison of three dipole matrix element formulations. $N$ denotes the number of occupied electrons and $N_{\mathrm{unocc}}$ denotes the number of unoccupied states to be included. Note that $\mathcal{O}(N_\text{}^3N_\text{unocc}) \simeq \mathcal{O}(N^4)$, since $N_\text{unocc}$ is proportional to $N$ in practical calculations.}
\centering
\begin{ruledtabular}
\begin{tabular}{@{} c c c c @{}}
\multicolumn{1}{c}{Eq.} &
\multicolumn{1}{c}{Expansion basis} &
\multicolumn{1}{c}{Big $\mathcal{O}$} &
\multicolumn{1}{c}{Ref.} \\
\colrule
\eqref{eq:GSB}           &  Ground-state (GSB) & $\mathcal{O}(N_\text{}^3N_\text{unocc})$   & \cite{Liang2017PRL,Liang2018PRB} \\
\eqref{eq:CHB}           & Core-hole (CHB) & $\mathcal{O}(N_\text{}^4)$   & \cite{Roychoudhury2023PRB} \\
\eqref{eq:newformulation}& Selection-rule-aware (SRB) & $\mathcal{O}(N_\text{}^3)$   & this work \\
\end{tabular}
\end{ruledtabular}
\label{tab:comp_time}
\end{table}

In this work, we present a more symmetry-transparent form.
The key idea is to expand the dipole operator $\hat{O}$ with a selection-rule-aware basis (SRB), which consists of the core state $\corestate$, the ``polarized core'' state $C\ket{\hat{O}\phi_\mathbf{c}}$, which will be referred to as \textit{o}\textbf{c} for the index, and the others $\{\ket{\chi_i}\}$. Here, $C$ is the temporary normalization factor such that $|C|^2\braket{\hat{O}\phi_\mathbf{c}|\hat{O}\phi_\mathbf{c}}=1$.
Note that the existence of SRB is ensured by the selection rule 
$\braket{\phi_\mathbf{c}|\hat{O}|\phi_\mathbf{c} }=0$ and the {Gram}-Schmidt orthogornalization, and the non-uniqueness for the choice of the other $\{\ket{\chi_i}\}$ does not affect the final result as will be clear later. 
Then, its annihilation action to KS Slater determinants, following the traditional sense, would require enormous occupied/unoccupied channel separation \cite{note:remark}, due to the loss of consistency with KS states.
For example, to represent $\hat{a}_{o\mathbf{c}}\ket{\InitialKSstate}$, we should return to the GSB with the unitary transform relation,
\begin{align}\label{eq:decomposition}
\polarizedcorestate&= \hat{\mathbf{1}}\polarizedcorestate\nonumber\\
&=\left[\sum_{n\in{\text{occ}}}\ket{\initialstate{\mathit{n}}}\bra{\initialstate{\mathit{n}}} + \sum_{m\in{\text{unocc}}}\ket{\initialstate{\mathit{m}}}\bra{\initialstate{\mathit{m}}}\right] \ket{\hat{O}\phi_\mathbf{c}},
\end{align}
and use the fact that only the occupied states $\ket{\initialstate{\mathit{n}}}$ can be annihilated with a different coefficient $\braket{\initialstate{\mathit{n}}|\hat{O}\phi_\mathbf{c}}$, and the unoccupied states $\ket{\initialstate{\mathit{m}}}$ read zero.  
Instead, we will handle this issue by carefully adopting the dipole selection rule and the fermionic statistics, so that terms creating double occupancy automatically vanish.
Starting with the canonical anticommutation relation \( \{\hat{a}^{\dagger}_{m},\hat{a}_{n}\}=\delta_{m,n} \), one can rewrite Eq. \eqref{eq:org_mat} as 
\begin{align}\label{eq:SRB}
\braket{ \Psi_{\mathrm{f}} | \hat{O} | \Initialstate }
&= \sum_{m,n} o_{m,n}^{(\text{SRB})}\, \braket{ \Phi_{\mathrm{f}} | \hat{a}_m^\dagger \hat{a}_n | \InitialKSstate } \nonumber\\[3pt]
&= \sum_{m,n} o_{m,n}^{(\text{SRB})}\, \braket{ \Phi_{\mathrm{f}} | \delta_{m,n} - \hat{a}_n \hat{a}_m^\dagger | \InitialKSstate }\nonumber \\[3pt]
&= - \sum_{m,n} o_{m,n}^{(\text{SRB})}\braket{ \Phi_{\mathrm{f}} | \hat{a}_n \hat{a}_m^\dagger | \InitialKSstate }, 
\end{align}
where ${o}^\text{(SRB)}_{m,\mathbf{c}}$ is expanded by the SRB, and $m$ and $n$ run for the whole state. 
For the third line, we used $\braket{\FinalKSstate|\InitialKSstate}=\braket{\FinalKSstate|\ncore|\InitialKSstate}=0$.
We now let creation operators act on both the bra and the ket, formally permitting double occupancy. 
Since the final state contains no core electron, we have
\begin{equation}
\braket{\FinalKSstate,\chi_n|\InitialKSstate,\chi_m}=0\quad\text{where  } \ket{\chi_n}\neq \corestate.\label{eq:SRB2}
\end{equation}
The remaining case with $m=\mathbf{c}$ also vanishes by fermionic nilpotency: 
$\hat a^\dagger_{\mathbf c}\ket{\Phi_\mathrm{i}}
=\hat a^\dagger_{\mathbf c}\hat a^\dagger_{\mathbf c}\hat a_{\mathbf c}\ket{\Phi_\mathrm{i}}=0$.
On the other hand, by construction of the SRB, we have
\begin{equation}
o_{m,\mathbf{c}}^{(\text{SRB})}=\braket{\chi_m|\hat{O}\phi_\mathbf{c}}=0 \quad \text{where  } \ket{\chi_m}\neq C\polarizedcorestate,\label{eq:SRB3}
\end{equation}
so that all the terms in Eq.~\eqref{eq:SRB} vanish except the one with $\ket{\chi_n}=\corestate$ and $\ket{\chi_m}=C\polarizedcorestate$.
Hence,
\begin{align}
\braket{ \Psi_{\mathrm{f}} | \hat{O} | \Initialstate }
&=-o_{o\mathbf{c},\mathbf{c}}^{(\text{SRB})}\braket{\FinalKSstate,\phi_\mathbf{c}|\InitialKSstate,C\hat{O}\phi_\mathbf{c}}\nonumber\\
&=-\Newmatrixelement, \label{eq:newformulation}
\end{align}
which is the desired result. For the last equation, the linearity of the determinant and 
\begin{align}
Co_{o\mathbf{c},\mathbf{c}}^\mathrm{(SRB)} &= C \braket{C\hat{O}\phi_\mathbf{c}|\hat{O}|\phi_\mathbf{c}}\nonumber\\&
=|C|^2\braket{\hat{O}\phi_\mathbf{c}|\hat{O}\phi_\mathbf{c}}=1
\end{align}
were used. 
It should be noted that Eq. \eqref{eq:newformulation} immediately recovers Eq. \eqref{eq:CHB} through the cofactor expansion along $\ket{\hat{O}\phi_\mathbf{c}}$ as a column vector (see Sec. \ref{sec:multiple_excited_state}), and is reduced to $\mathcal{O}(N^3)$ without any assumption. 


\begin{figure}
    \centering
  \includegraphics[width=\columnwidth]{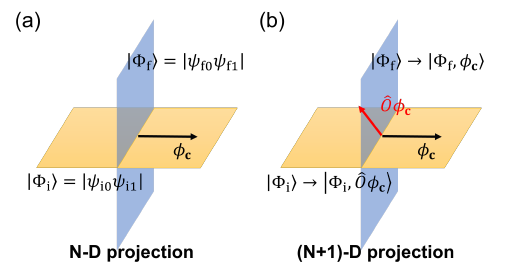}
    \caption{Geometrical interpretation of the matrix element between Slater determinants ($N=2$). The horizontal plane (yellow) stands for the initial state, spanned by $\initialstate{0} $ and $\initialstate{1}$, while the perpendicular plane (blue) for the final state spanned by $\finalstate{0}$ and $\finalstate{1}$. (a) shows $\braket{\FinalKSstate|\InitialKSstate}=\braket{\FinalKSstate|\ncore|\InitialKSstate}=0$, and (b) shows $\Newmatrixelement\neq0$ in general.}
    \label{fig:geo}
\end{figure}

According to Eq.~\eqref{eq:newformulation}, two geometric parameters $\theta$ and $\varphi$ govern the matrix element, which arise from the final and initial state, respectively.
Let us recall that the $N$-body Slater determinant is identified with the exterior product $\wedge^N\psi_{n}$ over the Hilbert space, i.e., the $N$-volume element spanned by $\{{\psi_n}\}$, preserving the inner product as a projection.
For the simplest case $N = 1$, $- \braket{\finalstate{},\phi_\mathbf{c}|\initialstate{},\hat{O}\phi_\mathbf{c}}$ can be interpreted as the projection of the $\ket{\initialstate{}, \hat{O}\phi_\mathbf{c}}$ subspace onto the $\ket{\finalstate{},\phi_\mathbf{c}}$ subspace and vice versa.
Then, since $\ket{\initialstate{}}=\corestate$ (for  $\braket{\ncore}=1$) and the two subspaces share the $\corestate$ vector, the projection is characterized by the angle $\theta$ between the remaining $\polarizedcorestate$ and $\ket{\finalstate{}}$, 
\begin{equation}
   - \braket{\finalstate{},\phi_\mathbf{c}|\phi_\mathbf{c},\hat{O}\phi_\mathbf{c}}=\braket{\finalstate{}|\hat{O}\phi_\mathbf{c}}\propto \cos{\theta}.
\end{equation}
Thus, the core excitation to $\ket{\finalstate{}}$ is maximized at $\theta=0$ and forbidden at $\theta=\pi/2$. This clarifies the role of $\theta$, which is defined similarly for larger $N$ (see below). 

Next, let us consider the case $N=2$ as shown in Fig. 1, in which each state corresponds to the subspace spanned by its own occupied states, up to a unitary rotation within the occupied subspace. 
Figure \ref{fig:geo}(a) illustrates the geometry relative to the core state $\corestate$: the $\ket{\InitialKSstate}$ subspace is aligned (parallel) with $\corestate$, while the $\ket{\FinalKSstate}$ subspace is orthogonal to it, so that $\braket{\Phi_{\mathrm f}|\Phi_{\mathrm i}}=\braket{\Phi_{\mathrm f}|\ncore|\Phi_{\mathrm i}}=0$.
By contrast, after adding $\polarizedcorestate$ and $\corestate$ to each subspace, respectively, the three-dimensional projection $-\Newmatrixelement$
in Fig.~\ref{fig:geo}(b) is generally nonzero and can be parameterized by two angles on $\polarizedcorestate$: 
$\theta$ with respect to $\ket{\FinalKSstate}$ and $\varphi$ with respect to $\ket{\InitialKSstate}$.
Physically, $\varphi$ encodes the Pauli exclusion principle in the polarized core state: 
Decompose $\polarizedcorestate=\ket{[\hat{O}\phi_\mathbf{c}]_\parallel}+\ket{[\hat{O}\phi_\mathbf{c}]_\perp}$ as Eq. \eqref{eq:decomposition}, so that $\braket{[\hat{O}\phi_\mathbf{c}]_\parallel|\hat{O}\phi_\mathbf{c}}\propto\cos{\varphi}$.
Since $\hat{a}^\dagger_{o\mathbf{c}}\ket{\InitialKSstate}$ automatically reads 
\begin{align}
\ket{\InitialKSstate,\hat{O}\phi_\mathbf{c}}&=\ket{\InitialKSstate,[\hat{O}\phi_\mathbf{c}]_\parallel}+\ket{\InitialKSstate,[\hat{O}\phi_\mathbf{c}]_\perp}\nonumber\\
&=\ket{\InitialKSstate,[\hat{O}\phi_\mathbf{c}]_\perp},
\end{align}
core excitation from $\ket{\InitialKSstate}$ is forbidden at $\varphi=0$ and maximized at $\varphi=\pi/2$.
This clarifies the role of $\varphi$, which is opaque in the GSB/CHB treatments \cite{Liang2017PRL,Liang2018PRB, Roychoudhury2023PRB}, and shows that we do not need to explicitly construct $\ket{[\hat{O}\phi_\mathbf{c}]_\perp}$ in the creation-based treatment.

For the larger $N$ case, the residual overlap between the remaining $\ket\InitialKSstate$ and $\ket\FinalKSstate$ excluding $\polarizedcorestate$  and $\corestate$ influences the matrix element and suppresses additional deep-state excitations by holding $\ket\FinalKSstate$ similar to $\ket{\InitialKSstate}$. 
Although it appears in GSB/CHB as many intermediate determinants, it can be reduced to $-\Newmatrixelement$ in the one-rank higher dimension, enabling efficient calculations.

\subsection{Multiple excited state}\label{sec:multiple_excited_state}
Our next goal is to calculate $|\Newmatrixelement|^2$ of many core-excited states, which are represented as the final KS Slater determinants with different configurations.
To make further simplification, we can continue to adopt the previous idea, which has been applied to the GSB/CSB treatments \cite{Liang2018PRB, Roychoudhury2023PRB}.
Let us denote the final state $\bra{\psi_{\mathrm{f},m}}$ and the core state $\bra{\phi_\mathbf{c}}$ as a $(N+1)$-dimensional row vector $\mathbf{a}_m\in \mathbb{C}^{N+1}$, which consists of the overlap with $\ket{\InitialKSstate}=|\{\psi_{\mathrm{i}}\}|$ and $\polarizedcorestate$,
\begin{align}
{\bf a}_m:=
\begin{bmatrix}
\braket{\finalstate{,\mathit{m}}|\initialstate{,1}}\quad... \quad\braket{\finalstate{,\mathit{m}}|\hat{O}\phi_\mathbf{c}}
\end{bmatrix}
,\\
\quad\quad\quad\quad(1\leq m \leq M)\nonumber
\end{align}
where $M$ is the total number of KS states (occupied $+$ unoccupied) and $\mathbf{a}_0$ keeps for $\bra{\phi_\mathbf{c}}$. 
Then, the dipole matrix element $\Newmatrixelement$ coincides with the determinant of the overlap matrix $A$ (e.g., for $\bra\FinalKSstate=|\{\mathbf{a}_0,...,\mathbf{a}_N\}|$),
\begin{equation}
 A=
\begin{bmatrix}\label{eq:definition_A}
\mathbf{a}_0 \\
\vdots \\
\mathbf{a}_{N}
\end{bmatrix},
 \quad \det A=\Newmatrixelement.
\end{equation}
Let us assume that $\{\mathbf{a}_m\}$ $({0\leq m\leq M})$ contains at least $N+1$ linear independent vectors. After reordering, place them as the first 
$N+1$ rows so that the associated matrix $A$ satisfies $\det A\neq0$.
In practice, it can be found within the lowest or single excitation configuration.
Then, $\{\mathbf{a}_n\}$ ${(0\leq n\leq N}$) spans the vector space $\mathbb{C}^{N+1}$ so that any ${\bf a}_m\in \mathbb{C}^{N+1}$ can be linearly expanded as
\begin{equation}
    \mathbf{a}_m=\sum_{0\leq n\leq N}K_{mn}\mathbf{a}_n \quad(0\leq m\leq M),
\end{equation}
or equivalently, 
\begin{align}\label{eq:matrix_K}
\begin{bmatrix}
\mathbf{a}_{0} \\
\vdots \\
\mathbf{a}_{M}
\end{bmatrix}
=
{{K}}
\begin{bmatrix}
\mathbf{a}_0 \\
\vdots \\
\mathbf{a}_{N}
\end{bmatrix}
\;\;\Longleftrightarrow\;\;
{{K}} = B A^{-1},
\end{align}
where $B$ stands for the $(M+1)\times (N+1)$ matrix in LHS, and $K_{mn}$ is the linear coefficient of $\mathbf{a}_m$ in $\mathbf{a}_n$.
Then, the matrix element of $\ket{\FinalKSstate^{(n\to m)}}$, which includes the single excitation ${\bf a}_{n}\to{\bf a}_{m}$, can be written as (e.g., for $n=N$) 
\begin{align}
\left|
\begin{array}{c}
\mathbf{a}_0 \\
\vdots \\
\mathbf{a}_m = \sum K_{mn} \mathbf{a}_n
\end{array}
\right|
= {{K}_{mN}}
\left|
\begin{array}{c}
\mathbf{a}_0 \\
\vdots \\
\mathbf{a}_N
\end{array}
\right|
= {{K}_{mN}}\det A.\label{eq:single_excitation}
\end{align}
For the first equation, the bilinear and alternating property of the determinant were used. 
Consequently, once $K$ is determined through Eq. \eqref{eq:matrix_K}, all matrix elements with a single excitation can be obtained through its component, allowing the simultaneous computations of the determinants for all single-excitation final states.

Although the benchmark calculations in this study are    limited to the above single excitation, we present two approaches for double excitations. The first is, as suggested in Ref. \cite{Liang2018PRB}, to change two rows in Eq. \eqref{eq:single_excitation}, reading
\begin{equation}
\left|
\begin{array}{c}
\mathbf{a}_0 \\
\vdots \\
\mathbf{a}_{m'}=\sum K_{m',n} \mathbf{a}_n\\
\mathbf{a}_m = \sum K_{mn} \mathbf{a}_n
\end{array}
\right|
=
\left|
\begin{array}{c}
K_{m',N-1}\; K_{m,N-1} \\
K_{m',N}\; K_{mN}
\end{array}
\right|
\det A,\label{eq:double_ex}
\end{equation}
which requires a determinant calculation for each configuration. 
Alternatively, we suggest the recursive approach by updating $(A,K)\to(A', K'=B[A']^{-1})$ with new basis $A'$ after single excitation, i.e., 
\begin{align}
\left|
\begin{array}{c}
\mathbf{a}_0 \\
\vdots \\
\mathbf{a}_{m'}\\
\mathbf{a}_m
\end{array}
\right|
=K'_{m',N-1}\det A',
\label{eq:double_ex2}
\end{align}
where $\det A'=K_{mN}\det A $ is previously given in Eq. \eqref{eq:single_excitation}. 
Note that Eq. \eqref{eq:double_ex2} holds if $K_{mN}\neq0$.
By comparing Eq. \eqref{eq:double_ex} with Eq. \eqref{eq:double_ex2}, we find the recursive relation between $K'$ and $K$,
\begin{align}
   K'_{m',N-1}&=
\left|
\begin{array}{c}
K_{m',N-1}\; K_{m,N-1} \\
K_{m',N}\; K_{mN}
\end{array}
\right|[K_{mN}]^{-1} \nonumber\\
&= K_{m',N-1} - \frac{K_{m,N-1}K_{m',N}}{K_{mN}},
\end{align}
while $K'$ will be directly calculated through $B[A']^{-1}$.
Since Eq. \eqref{eq:double_ex2} enables the collective evaluation analogous to the single excitation case, it should be carefully implemented and tested, while monitoring duplicate counting and convergence behavior.

\section{Implementation}
\subsection{General}
\begin{table}[tbp]
\caption{Basis functions used for each X-ray absorption calculation.}
\centering
\begin{ruledtabular}
\begin{tabular}{@{}ll@{}}
\multicolumn{1}{c}{\textbf{Element}} &
\multicolumn{1}{c}{\textbf{Basis functions}} \\
\midrule
\multicolumn{2}{l}{\textbf{C K-edge}} \\
\quad C  & \texttt{C7.0.1s-s4p3d2} \\
\quad C\_CH  & \texttt{C7.0.1s\_CH-s4p3d2} \\
\quad H  & \texttt{H7.0-s3p2} \\
\addlinespace
\multicolumn{2}{l}{\textbf{B K-edge}} \\
\quad B  & \texttt{B7.0.1s-s3p2d1} \\
\quad B\_CH  & \texttt{B7.0.1s\_CH-s3p2d1} \\
\quad N  & \texttt{N6.0-s2p2d1} \\
\quad Mg & \texttt{Mg7.0-s3p2d2} \\
\addlinespace
\multicolumn{2}{l}{\textbf{O K-edge}} \\
\quad O  & \texttt{O7.0.1s-s4p3d2} \\
\quad O\_CH  & \texttt{O7.0.1s\_CH-s4p3d2}\qquad\qquad\qquad\qquad \\
\quad Li & \texttt{Li8.0-s3p2} \\
\quad C  & \texttt{C6.0-s2p2d1} \\
\addlinespace
\multicolumn{2}{l}{\textbf{Li K-edge}} \\
\quad Li & \texttt{Li8.0.1s-s3p2d1} \\
\quad Li\_CH & \texttt{Li8.0.1s\_CH-s3p2d1} \\
\quad O  & \texttt{O6.0-s2p2d1} \\
\quad C  & \texttt{C6.0-s2p2d1} \\
\end{tabular}
\end{ruledtabular}\label{tab:basis}
\end{table}

All calculations were performed with \textsc{OpenMX} \cite{OpenMX_Website, Ozaki2005_PRB_ProjectorExpansion, Duy2014_CPC_3D_DomainDecomp, Lejaeghere2016_Science_Reproducibility}, which employs norm-conserving pseudopotentials (PPs) \cite{Theurich2001_PRB_SOC_PSP,Morrison1993_PRB_Vanderbilt_Hermitian} and a linear combination of pseudo-atomic orbitals (LCPAO) \cite{Ozaki2003_PRB_VariationalNAO}.
The basis functions used are summarized in Table~\ref{tab:basis}. 
In labels such as \texttt{H7.0-s3p2d1}, the leading letter denotes the element; the number (e.g., 7.0) is the confinement (cutoff) radius in Bohr used when generating the basis with \textsc{ADPACK} \cite{OpenMX_Website}; and the trailing string \texttt{s3p2d1} specifies the numbers of optimized radial functions for the \(s\), \(p\), and \(d\) channels, respectively. 
A suffix \texttt{\_1s} in the basis (e.g., \texttt{C7.0\_1s-s4p3d2}) indicates that the 1$s$ core state is retained, while \texttt{\_CH} marks a core-hole–adapted variant used for final state calculations. 
Their precision was cross-checked by benchmarking X-ray photoemission spectroscopy (XPS) binding energies \cite{Ozaki2017PRL, Lee2017_Silicene_CoreStates_PRB,
      Yamazaki2018_JPCC_SinglePtGraphene,
      Lee2018_PRB_BoropheneBonding,
      Kurumada2024_ChemEurJ_AlIAnionDimer,
      ChemPhysChem2023_H2Vacancies_MoS2_APXPS,
      Sato2025_PRM_3x3Si_on_Al111,
      PCCP2024_PdCu111_HCOOH_H2_Process}. 
Exchange–correlation functional was treated within the generalized gradinent approximation (GGA) by Perdew, 
Burke, and Ernzerhof \cite{Perdew1996_PRL_PBE}, and an electronic temperature of 300 K was used for Fermi–Dirac occupations. To capture core relaxation that breaks translational symmetry, we adopted a supercell approach with $\Gamma$-point sampling. Core excitations were computed by an explicit DFT-$\Delta$SCF framework augmented with a penalty projector (Eq. \eqref{eq:ks-penalty}). The threshold penalty constant $\Delta$ which reverses the energy ordering between the core-occupied and core-excited subspaces (Eq.~\eqref{eq:theorem2_proof}) is evaluated by the core-level KS eigenvalue; in our benchmarks,
$\Delta = 100$ Ry ($\simeq$ 1361 eV) was found sufficient for all
systems tested here.
Real-space grid cutoff was chosen to converge total energies to within $\sim$ 0.1–0.2~eV. 
\subsection{KS equation with LCPAO} \label{section:LCAO}

In the LCPAO representation, the KS state at a wave vector \( \mathbf{k} \) is expanded as
\begin{equation}\label{eq:LCAO}
  \ket{\psi^{\mathbf{k}}_{\alpha,m}}=\sum_\mu C_{\alpha,\mu m }(\mathbf{k})\ket{\phi^{\mathbf{k}}_\mu},
\end{equation}
where $\alpha\in\mathrm{\{i,f\}}$ labels initial/final states, \( m \) is the band index, and $\{\ket{\phi^{\mathbf{k}}_\mu}\}$ are Bloch-summed PAOs.
The PAO overlap and dipole matrices are
\begin{align}
S_{\mu\nu}({\mathbf{k}})&=\braket{\phi^{\mathbf{k}}_\mu|\phi^{\mathbf{k}}_\nu},\label{eq:overlap}
\\
P^{\boldsymbol\epsilon}_{\mu\nu}({\mathbf{k}})&=\braket{\phi^{\mathbf{k}}_\mu|\hat {O}|\phi^{\mathbf{k}}_\nu}\label{eq:dipole},
\end{align}
where $\boldsymbol\epsilon$ explicitly denotes the polarization direction of $\hat{O}$.
Each entity is calculated by momentum-space and real-space integration, respectively. 
Then, the projection of the KS equation (Eq. \eqref{eq:ks-penalty}) onto each PAO basis $\ket{\phi^{\mathbf{k}}_\mu}$ yields the generalized eigenproblem ($\mathbf{k}$ is hidden for simplicity),
\begin{equation}
\begin{aligned}
  &\sum_{\nu}\!\Bigl[ H_{\mu\nu} + \Delta S_{\mu \mathbf{c}}S_{\mathbf{c}\nu} \Bigr] C_{\nu m}
  =\sum_{\nu} S_{\mu\nu} C_{\nu m} \epsilon_{m} ,
  \end{aligned}
\end{equation}
where $H_{\mu\nu}({\bf k})=\braket{\phi^{\mathbf{k}}_\mu|\hat h_\text{KS}|\phi^{\mathbf{k}}_\nu}$. 
Since $\corestate$ is fixed, the penalty matrix $\Delta\, S_{\mu \mathbf{c}}(\mathbf{k})\,S_{\mathbf{c}\nu}(\mathbf{k})$ remains unchanged throughout the SCF iterations.

Although the benchmark calculations are limited to the $s$ core state (K-edge) with collinear Hamiltonian in this study, $p$/$d$-core state excitations (L/M-edge)  with non-collinear Hamiltonian can be handled by promoting all quantities to spinor form with the corresponding core state as
\begin{align}
    \corestate&=\ket{R_\mathbf{c}{\Phi_J^M}},
\end{align}
where 
\begin{equation}\label{eq:spinor1}
    \ket{\Phi_J^M}= \sqrt{\frac{l+m+1}{2 l+1}}\ket{Y_l^m \uparrow}+\sqrt{\frac{l-m}{2 l+1}}\ket{Y_l^{m+1} \downarrow}, 
\end{equation}
for $J=l+\frac{1}{2}$ and $M=m+\frac{1}{2}$, and 
\begin{equation}\label{eq:spinor2}
    \ket{\Phi_J^M}= \sqrt\frac{l-m+1}{2 l+1}\ket{Y_l^{m-1} \uparrow}-\sqrt\frac{l+m}{2 l+1}\ket{Y_l^{m} \downarrow}, 
\end{equation}
for $J=l-\frac{1}{2}$ and $M=m-\frac{1}{2}$. Here, $\ket{R_\mathbf{c}}$ and $\ket{\Phi_J^M}$ are the radical and spherical component, respectively, $\ket{\uparrow}$ and $\ket{\downarrow}$ are spin states, and each coefficient $\sqrt{\cdots}$ in Eqs. \eqref{eq:spinor1} and \eqref{eq:spinor2} is the Clebsch-Gordan coefficient.
The relevant implementations and calculations based on the core-explicit $\Delta$SCF framework are also found in Refs. \cite{Ozaki2017PRL, OpenMX_Website}.

During the collinear calculations, we observed that the lowest core-excited state is often an optically dark state, which yields no absorption because of the discrepancy of the spin moment. 
This is because the dipole operator $\hat{O}$ preserves the spin moment $\hat{S}_\mathrm{z}$, i.e., $[\hat{O},\hat{S}_\mathrm{z}]=0$, so
\begin{gather}
    \braket{\Finalstate|[\hat{O},\hat{S}_\mathrm{z}]|\Initialstate}=(m_\mathrm{i}-m_\mathrm{f})\dipolematrix=0 \label{eq:collinear} \nonumber\\
    \Longrightarrow \nonumber \\
    \text{if}\; m_\mathrm{f}\neq m_\mathrm{i}, \quad\dipolematrix=0, \label{eq:collinear2}
\end{gather}
where 
\begin{equation}
\hat{S}_\mathrm{z}\ket{\Psi_{\alpha}}=m_{\alpha}\ket{\Psi_{\alpha}}\quad (\alpha\in\{\mathrm{i, f}\}).\label{eq:collinear3}
\end{equation}
It should be noted that Eq. \eqref{eq:collinear3} only holds with the collinear Hamiltonian which satisfies $[\hat{H},\hat{S}_\mathrm{z}]=0$, so that the eigenstate of the Hamiltonian can also be the eigenstate of the spin moment. 
To obtain the absorption-yield state, we further required the conservation of the spin moment during self-consistent iteration by constraining the spin populations. 
Note that this constraint is also ensured by the GL theorem relevant to the spin moment $\hat{S}$ in addition to the constraint for the core-hole creation.

\subsection{Dipole matrix element\label{implementation2}}
The {overlap} matrix $A$ and $B$ (Eqs. \eqref{eq:definition_A} and \eqref{eq:matrix_K}) can be decomposed into four blocks:
\[
  A \text{ or }B\;=\;
  \begin{bmatrix}
    \text{(1,1)} & \text{(1,2)}\\[2pt]
    \text{(2,1)} & \text{(2,2)}
  \end{bmatrix},
\]
where these blocks are defined by \\ \noindent
{(1,1) initial vs. final}
\begin{align}\label{eq:matrix_element1}
\braket{\finalstate{,\mathit{m}}^\mathbf{k}|\initialstate{,\mathit{n}}^\mathbf{k}}
&= \sum_{\mu,\nu} C_{\mathrm{f},m\mu}^{*}\, S_{\mu \nu}\, C_{\mathrm{i},\nu n}\nonumber\\
&= [C^{*}_\mathrm{f}S C_\mathrm{i}]_{mn},
\end{align}

\noindent
{(1,2) polarized core vs. final}
\begin{align}\label{eq:matrix_element2}
\braket{\finalstate{,\mathit{m}}^\mathbf{k}|[\hat{O}\phi_\mathbf{c}]^\mathbf{k}} &= \sum_\mu C^{}_{\mathrm{f},m\mu}P^{\boldsymbol{\epsilon}}_{\mu\mathbf{c}}\nonumber\\ &= [C^*_\mathrm{f}P^{\boldsymbol{\epsilon}}]_{m\mathbf{c}},
\end{align}

\noindent
{(2,1) initial vs. core}
\begin{align}\label{eq:matrix_element3}
  \braket{\phi_{\mathbf c}^\mathbf{k}|\psi_{\mathrm{i},n}^\mathbf{k}}
  &= \sum_\mu S_{\mathbf{c}\mu}C_{\mathrm{i},\mu n}\nonumber\\ &= [SC_\mathrm{i}]_{\mathbf{c}{n}},
\end{align}

\noindent
{(2,2) polarized core vs. core}
\begin{equation}\label{eq:matrix_element4}
  \braket{\phi_{\mathbf c}^\mathbf{k}|[\hat {O} \phi_{\mathbf c}]^\mathbf{k}}
  = 0.
\end{equation}
As shown in Eqs. \eqref{eq:matrix_element1}--\eqref{eq:matrix_element4}, each block can be obtained by multiplication of the fundamental matrices that appeared in Eqs. \eqref{eq:LCAO}--\eqref{eq:dipole}, which ensures efficient parallel computations. 
It is also possible to adopt different basis sets in initial (i) and final (f) calculations by promoting $S_{\mu\nu}$ to $S_{\mathrm{f}\mu, \mathrm{i}\nu}$, which is overlap between two basis sets.

\begin{table}[t]
\caption{Expectation value of $\ncore$ of initial (ground) and final (lowest core-excited) state and its inner product of test materials.}
\label{tab:example}
\begin{ruledtabular}
\begin{tabular}{lccc}
Core state   & $\braket{\InitialKSstate|\hat n_\mathbf{c}|\InitialKSstate}$        & $\braket{\FinalKSstate|\hat n_\mathbf{c}|\FinalKSstate}$         & $|\braket{\FinalKSstate|\InitialKSstate}|^2$    \\
\hline
\multicolumn{4}{l}{\textbf{C 1s}}\\
\quad C$_2$H$_2$   & 0.9986  & 0.0001  & $<10^{-4}$ \\
\quad C$_2$H$_6$   & 0.9986  & 0.0005  & $0.0001$ \\
\multicolumn{4}{l}{\textbf{B 1s}}\\
\quad BN           & 0.9941  & 0.0004 & $<10^{-4}$ \\
\quad MgB$_2$      & 0.9932  & 0.0005  & $<10^{-4}$ \\
\multicolumn{4}{l}{\textbf{O 1s}}\\
\quad Li$_2$O$_2$  & 0.9978  & 0.0002  & $<10^{-4}$ \\
\quad Li$_2$CO$_3$ & 0.9986  & 0.0001  & $<10^{-4}$ \\
\quad Li$_2$O      & 0.9986  & 0.0001  & $<10^{-4}$ \\
\multicolumn{4}{l}{\textbf{Li 1s}}\\
\quad Li$_2$O$_2$  & 0.9992  &  0.0002  & 0.0002 \\
\quad Li$_2$CO$_3$ & 0.9922  &  0.0001   & 0.0002 \\
\quad Li$_2$O      & 0.9918 & 0.0002    & 0.0001 \\
\end{tabular}
\end{ruledtabular}\label{tab:inner_product}
\end{table}

\section{Results}
The internal validation tests in Sec.~\ref{sec:method_verification}
show that the constrained-search construction produces the desired
core-excited states with the intended core occupation and that the SRB
formulation exhibits the effective $\mathcal{O}(N^3)$ scaling, in contrast to
the $\mathcal{O}(N^4)$ cost of the CHB treatment. The case studies in
Sec.~\ref{sec:case_study} then probe sensitivity to local chemistry and
orbital character in molecular systems of XAS, as well as band-structure and symmetry
effects in solids. In all systems considered, we compare the calculated
spectra with reference data and analyze polarization-resolved trends.

\subsection{Method verification}\label{sec:method_verification}

Table~\ref{tab:inner_product} provides a quantitative validation of the
proposed construction of the core-excited state. The ground-state expectation
values $\braket{\InitialKSstate|\hat n_\mathbf{c}|\InitialKSstate}=\sum_{m\in \mathrm{occ}}|\braket{\phi_\mathbf{c}|\psi_{\mathrm{i},m}}|^2$ remain very close
to unity ($0.99$--$1.00$), indicating that the targeted $1s$ core orbital is
essentially fully occupied in the initial state, whereas the lowest
core-excited states exhibit expectation values on the order of $\sim 10^{-4}$,
demonstrating that the core hole is well localized on the designated atom.
The small but finite residual occupations can be attributed to weak
hybridization with neighboring atoms, but their magnitudes are uniformly
negligible. The squares of their inner product $|\braket{\FinalKSstate|\InitialKSstate}|^2$ are
all on the order of $\sim 10^{-4}$, showing that the excited state is
effectively orthogonal to the ground state, even though the individual KS
states of the two calculations are not mutually orthogonal. The robustness
and uniformity of these indicators across all systems considered demonstrate
that the explicit-core $\Delta$SCF framework reliably enforces the targeted
core occupation and yields well-defined core-excited states largely
independent of the chemical environment.

Next, to assess the computational efficiency of the SRB formulation relative to the
CHB treatment, we benchmarked both under identical conditions on a single CPU
core using LAPACK’s LU factorization routine \texttt{zgetrf} (complex double;
threading disabled). As shown in Fig.~\ref{fig:Comp_time}, the SRB scheme
follows the expected $\mathcal{O}(N^3)$ scaling of LU factorization (green
curves), as verified by wall-clock measurements for randomly generated
$\mathbb{C}^{N\times N}$ matrices; the material data points for c-\ce{BN}
($N=320$), \ce{Li2O2}, \ce{Li2O}, and \ce{MgB2} (black dots) fall on the same
trend. In contrast, the CHB formulation expresses the matrix element as a sum
over $N$ occupied states. If each term requires a
comparable $\mathcal{O}(N^3)$ LU factorization, the total cost scales as
$\mathcal{O}(N^4)$. Accordingly, the purple curves in
Figure~\ref{fig:Comp_time} shows CHB estimates obtained by multiplying single-term cost with $N$, and the gray dots mark the corresponding
material points. As seen in the left panel (sub-second regime) and the right
panel (extended to minutes), the CHB evaluation becomes computationally demanding for $N \gtrsim 10^4$, exceeding one hour per excited state,
whereas SRB remains computationally tractable over the same range, completing
in approximately one second.
These results substantiate the favorable scaling and practical
efficiency of the SRB formulation.

\begin{figure}[t]
  \centering

  \includegraphics[width=\columnwidth]{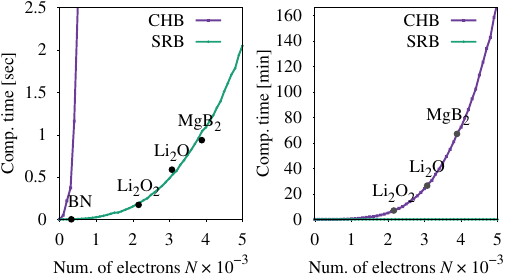}
\caption{Computational time to evaluate a single dipole matrix element using CHB (Eq.~\eqref{eq:CHB}) and SRB (Eq.~\eqref{eq:newformulation}) formulations (Left: seconds (zoomed). Right: minutes (extended)). Measurements were performed on a single CPU core using the LAPACK LU factorization routine (\texttt{zgetrf}). SRB (green) values are measured wall-clock times from runs on a priori randomly generated \(\mathbb{C}^{N\times N}\) matrices, whereas CHB (purple) values are estimated by multiplying single-term cost with $N$ to account for the summation. Black dots indicate exact values of representative materials with SRB (averaged over 10 runs), whereas gray dots indicate estimation with CHB for comparison.}\label{fig:Comp_time}
\end{figure}

\subsection{Case study}\label{sec:case_study}
Figures \ref{fig:C_K-edge}(c,d), \ref{fig:B_K-edge}(g,h), and \ref{fig:O_K-edge}(d--k) compare our calculated spectra with experimental references. In the calculations, red/green/blue curves denote linear $\boldsymbol{\hat{{\epsilon}}}= \hat{\mathrm{x}}/\mathrm{y}/\hat{\mathrm{z}}$ polarizations; the black solid curve is the orientational average, and the black dashed    curve shows the experimental data.  To emulate lifetime and instrumental effects, we apply a constant Gaussian broadening of 0.5 eV; stick spectra are shown to indicate discrete transition energies and intensities. Absolute energy alignment is taken as the difference between the lowest core-excited and ground-state total energies, without any empirical shift.



\begin{figure}[t]
  \centering

  \includegraphics[width=\columnwidth]{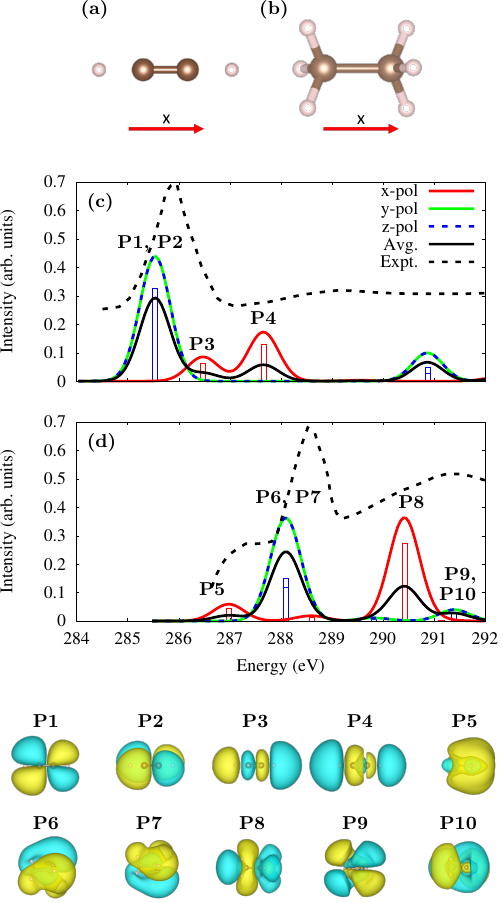}

  \caption{(a) Linear acetylene (\ce{C2H2}) and (b) ethane (\ce{C2H6}) with the incident polarization indicated (\ce{C}: brown, \ce{H}: bright brown). (c) and (d) show the corresponding carbon K-edge XAS; polarization-resolved and averaged spectra are plotted with peak labels P1–-P10. Experimental reference spectra (dashed line) were digitized from published figures (Ref.~\cite{BesleyNoble2007JPCC}) and replotted for comparison. The lower panels display isosurfaces of the final Kohn--Sham states corresponding to peaks P1–P10, with the core hole located on the right C atom; P1--P2, P6--P7, and P9--P10 are degenerate. Yellow and cyan denote opposite phases of the eigenstate.}
  \label{fig:C_K-edge}
\end{figure}

\begin{figure}[t]
  \centering

  \includegraphics[width=\columnwidth]{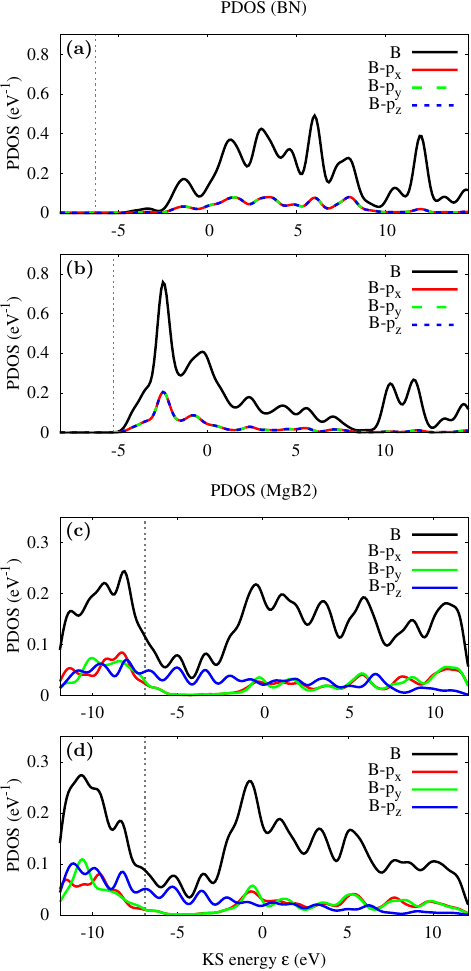}

\caption{%
Projected density of states (PDOS) of the B atom which retains the 1$s$ core hole. Panels (a,b) show \ce{BN} and (c,d) show \ce{MgB2}; (a,c) are ground state (initial), while (b,d) include a B–1$s$ core hole (final). Each dashed line denotes the Fermi level.}
  \label{fig:PDOS}
\end{figure}

\begin{figure*}[t]
  \centering

  \includegraphics[width=\textwidth]{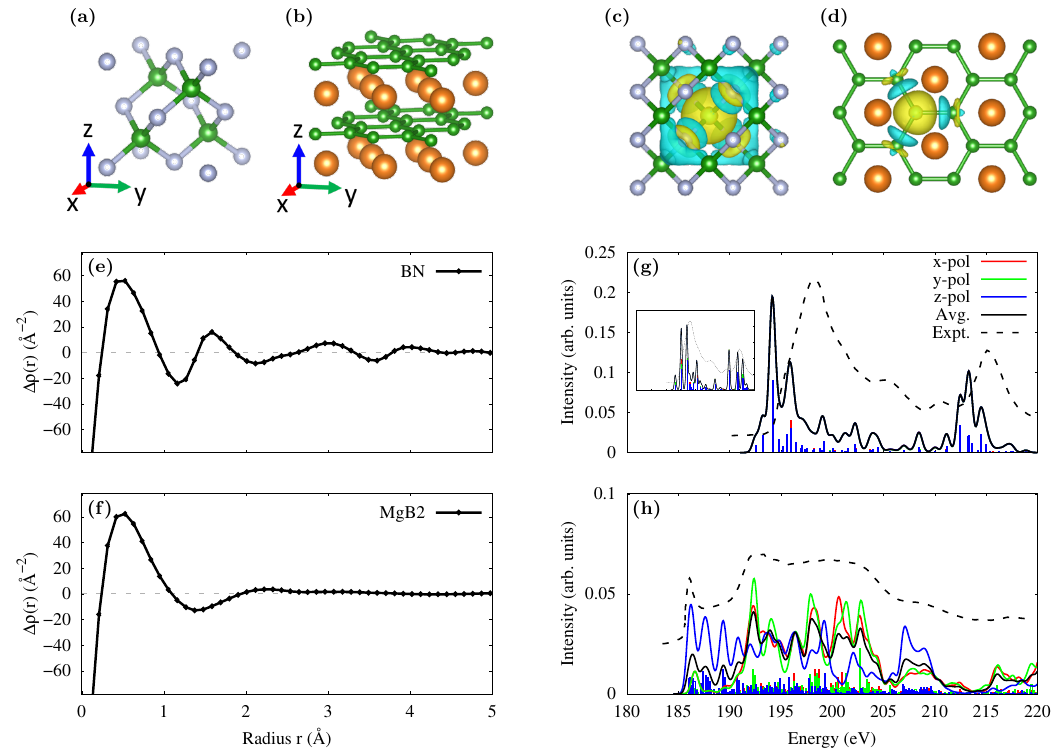}

  \caption{
  (a,b) Crystal structures of (a) \ce{c-BN} and (b) \ce{MgB2} (\ce{B}: green, \ce{N}: gray, Mg: \ce{orange}). (c,d) Isosurface of the density difference between final (core-excited) and initial (ground) state calculations of (c) \ce{BN} and (d) \ce{MgB_2}. Isolevel = $+$0.003 (cyan) and $-$0.003 (yellow). 
  (e,f) Spherically averaged density difference from the core hole center.    
  (g,h) B K-edge XAS of (g) c-\ce{BN} and (h) \ce{MgB2}. The inset in (g) shows the result from a smaller unit cell ($N=320$), while the main trace corresponds to a large supercell ($N=2560$). 
  Experimental reference spectra (dashed line) were digitized from published figures (Refs. \cite{Jayawardane2001PRB, Zhu2002PRL}) and replotted for comparison.
  }

  \label{fig:B_K-edge}
\end{figure*}


\begin{figure}[t]
  \centering

  \includegraphics[width=\columnwidth]{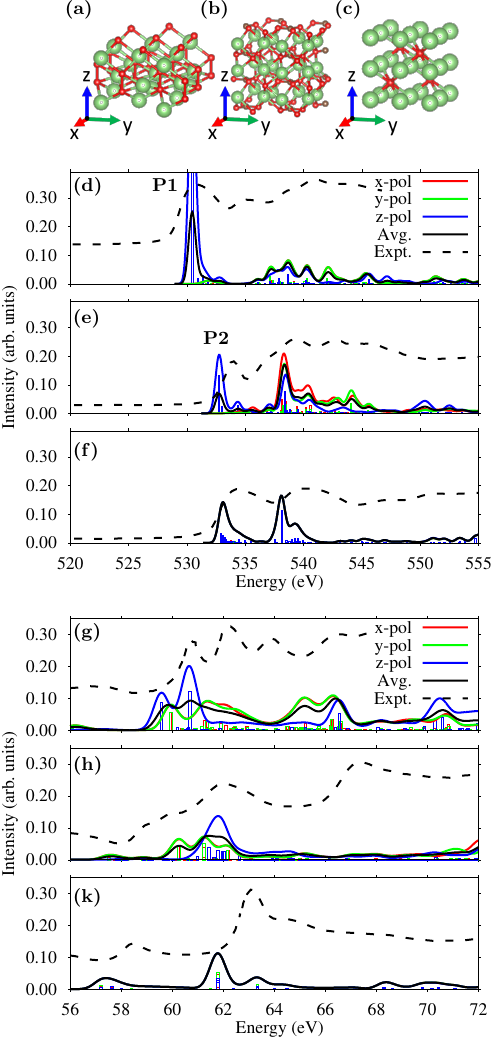}

\caption{(a--c) Crystal structures of (a) \ce{Li2O2}, (b) \ce{Li2CO3}, and (c) \ce{Li2O} (\ce{Li}: green, \ce{O}: red, \ce{C}: brown). (d--f) O K-edge and (g--k) Li K-edge XAS for the structures in (a--c), respectively. Experimental reference spectra (dashed line) were digitized from published figures (Ref. \cite{Qiao2012PLOS}) and replotted for comparison.}

  \label{fig:O_K-edge}
\end{figure}

\begin{figure}[b]
\centering

  \includegraphics[width=\columnwidth]{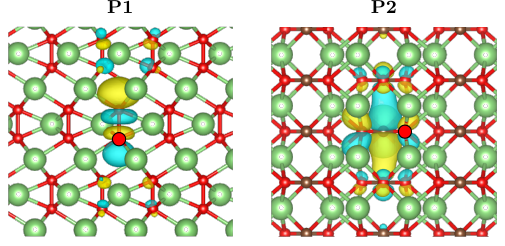}

\caption{Isosurfaces of the final KS states corresponding to peaks P1 (left)
and P2 (right) in Fig.~\ref{fig:O_K-edge}.  
The core hole is located on the highlighted O atom (red circle).  
P1 has $\sigma^*(\ce{O-O})$ character, while P2 has $\pi^*(\ce{C=O})$ character. Yellow and cyan denote opposite phases of the eigenstate (Isolevel $=\pm0.03$).}
\label{fig:O_isosurface}
\end{figure}

\paragraph*{\textbf{\upshape Carbon molecules (C K-edge).}}

For small carbon molecules, the acetylene–ethane pair provides a clean \(\pi\) vs. \(\sigma\) contrast under polarization resolution. With the molecular axis taken as \(\hat{\mathrm{x}}\), the polarized core state $\polarizedcorestate$ 
implies \(\pi^\ast(p_y,p_\mathrm{z})\) transitions for \(\boldsymbol{\hat\epsilon}\!\perp\!\hat{\mathrm{x}}\) and \(\sigma^\ast(p_\mathrm{x})\) for \(\boldsymbol{\hat\epsilon}\!\parallel\!\hat{\mathrm{x}}\).
Acetylene (C\(_2\)H\(_2\), \(N=14\); Fig.~\ref{fig:C_K-edge}(a,c)) is a prototypical linear \(\pi\) system: the cylindrical twofold degeneracy of \((p_\mathrm{y},p_\mathrm{z})\) yields strong \(\hat{\mathrm{y}}\), \(\hat{\mathrm{z}}\) absorption and a suppressed \(\hat{\mathrm{x}}\) response. The first dominant peak is unambiguously \(\pi^\ast\) (Fig.~\ref{fig:C_K-edge}(P1,P2)), while higher-energy features progressively gain \(\sigma^\ast\) admixture (Fig.~\ref{fig:O_K-edge}(P3,P4)).
By contrast, ethane (C\(_2\)H\(_6\), \(N=18\); Fig.~\ref{fig:C_K-edge}(b,d)) is a saturated \(\sigma\)-bonded molecule whose edge onset is already \(\sigma^\ast\)-dominated, and multiple peaks arise from several near-degenerate \(\sigma^\ast\) states. Among the two \(\sigma^\ast\) states in Figs.~\ref{fig:C_K-edge}(P5) and (P8), P5 lies at lower energy but shows a smaller absorption intensity, consistent with its more hole-side localized and $s$-symmetric shape, respectively.
In both molecules, the main-peak positions agree with experiment within \(|\Delta E|\lesssim 0.5\)~eV, supporting shift-free absolute alignment. The \(\hat{\mathrm{x}}/\hat{\mathrm{y}}/\hat{\mathrm{z}}\) decomposition thus exposes the symmetry–selection-rule contrast and cleanly separates the \(\pi^\ast\)-dominated case (C\(_2\)H\(_2\)) from the \(\sigma^\ast\)-dominated case (C\(_2\)H\(_6\)).

\medskip 
\noindent

\paragraph*{\textbf{\upshape Boron compounds (B K-edge).}}
In solids, \(\mathbf{k}\)-dispersion mixes transition channels, making configuration-resolved features less separable than in molecules. Equally important, the presence or absence of a band gap governs the relaxation (screening) of a core hole. This contrast is visible first in the PDOS (Figs.~\ref{fig:PDOS}(a-d)): in c-\ce{BN} (insulator, $N=2560$), the attractive potential of the B–1s core hole draws density of states to lower energies (Figs.~\ref{fig:PDOS}(a,b)); in \ce{MgB2} (metal, $N=3888$), states near  Fermi level \(E_\mathrm{F}\) efficiently screen the core hole, so the redistribution in the vicinity of \(E_\mathrm{F}\) remains comparatively small (Fig.~\ref{fig:PDOS}(c,d)). The density difference between initial (ground) and final (lowest core-excited) states, which visualizes the core hole and electron response up to the phase, is consistent with this picture (Figs.~\ref{fig:B_K-edge}(c-f)): in c-\ce{BN}, the isosurface of the density difference (isolevel$=\pm 0.003$) forms nearly isotropic alternating positive/negative shells around B and the spherical average \(\Delta\rho(r)\) with the core-hole center exhibits oscillations that persist out to \(\sim\!5~\text{\AA}\) (Figs.~\ref{fig:B_K-edge}(c,e)), whereas in \ce{MgB2} the response is anisotropic with protrusions extending along the \(c\)-axis toward \ce{Mg} layers and \(\Delta\rho(r)\) decays within \(\sim\!3~\text{\AA}\) (Figs.~\ref{fig:B_K-edge}(d,f)). These behaviors succinctly reflect the gap-controlled difference in screening between an insulator and a metal. We stress that PDOS alone cannot reproduce an absorption spectrum, due to the absence of the dipole selection rules (e.g., constraint spin channel and $\polarizedcorestate$ orbital dependency) and many-body effects; PDOS is therefore qualitative evidence of the available final-state character rather than a predictor of intensities.

Turning to XAS, Fig.~\ref{fig:B_K-edge}(g) shows the finite-cell effects with different supercell sizes. The inset ($N=320$) shows discretized peaks arising from finite-cell quantization, and the insufficient screening in too-small cells further distorts line shapes; increasing \(N\) mitigates both issues and smooths the spectrum toward the experimental envelope in the main panel ($N=2560$). On the other hand, polarization trends between two B compounds are also notable: c-\ce{BN} shows no resolvable anisotropy, consistent with cubic symmetry, while \ce{MgB2} exhibits clear anisotropy reflecting its layer structure; out-of-plane \(p_\mathrm{z}\)-like states dominate the onset while in-plane \(p_\mathrm{x,y}\) weight grows at higher energies. As a result, orientational averages recover the experimental line shapes and relative intensities once supercell-size convergence is achieved.

\medskip
\paragraph*{\textbf{\upshape Li oxides (O, Li K-edge)}}
We next turn to lithium oxides, which provide a controlled setting to test various charge states under different chemical environments. The three phases considered here differ in local O coordination and Li environment, offering distinct signatures at the O K-edge (unoccupied O-2$p$-hybridized states) and at the Li K-edge (Li-2$s$/2$p$ contributions). We focus on two metrics: (i) relative energy alignment across the three phases and between O and Li edges, and (ii) line-shape, including the onset and the first two peak intensities. As shown below, the calculated trends capture the phase-dependent shifts and anisotropies without any empirical energy shift, indicating that the present scheme is robust across chemically related oxides.
Figures \ref{fig:O_K-edge}(d)--(f) and \ref{fig:O_K-edge}(g)--(k) present the O and Li K-edge spectra, respectively, alongside the crystal structures shown in Figs. \ref{fig:O_K-edge}(a)--(c). 
For \ce{Li2O2}
($N = 2106$), the O K-edge (Figs. \ref{fig:O_K-edge}(a,d)) exhibits a sharp
leading edge that is strongest for $\hat{\mathrm{z}}$-polarization, consistent
with excitation into the antibond $\sigma^\ast(\ce{O-O})$ (Fig. \ref{fig:O_isosurface}(P1)). For
\ce{Li2CO3} ($N = 2016$) (Figs. \ref{fig:O_K-edge}(b,e)), a weak initial feature
near $\sim$ 533 eV has $\pi^\ast(\ce{C=O})$ character (Fig. \ref{fig:O_isosurface}(P2)) and is relatively
enhanced for $\hat{\mathrm{z}}$ polarization.
Both P1 and P2 correspond to core-hole-trapped excitonic states, in which the excited electron remains strongly bound to the O 1$s$ core hole in the final state. 
This bound core-exciton character accounts for the sharp leading-edge intensity and its pronounced polarization dependence.
By contrast, cubic \ce{Li2O} ($N=3072$) (Figs.~\ref{fig:O_K-edge}(c,f)) is nearly isotropic with negligible polarization dependence and shows two broad features at $\sim\!534$ and $\sim\!539~\mathrm{eV}$; although the calculated peaks are narrower, the overall two-peak motif agrees well with experiment.

For the Li K-edge, \ce{Li2O2} (Fig.~\ref{fig:O_K-edge}(g)) shows a pronounced anisotropy: a sharp onset dominated by \(\hat{\mathrm{z}}\)-polarization near \(60~\mathrm{eV}\) and $61~\mathrm{eV}$, followed by a weaker feature around \(66\text{--}68~\mathrm{eV}\). The orientational average reproduces the experimental peak positions but remains narrower. In \(\ce{Li2CO3}\) (Fig.~\ref{fig:O_K-edge}(e)), a modest \(\hat{\mathrm{z}}\)-enhanced feature appears at \(61\text{--}62~\mathrm{eV}\) with a broad shoulder above \(65~\mathrm{eV}\). The relative spacing and polarization hierarchy agree with experiment, while the computed bandwidth is underestimated. By contrast, cubic \(\ce{Li2O}\) (Fig.~\ref{fig:O_K-edge}(f)) is nearly isotropic, exhibiting a single dominant near-edge peak at \(62\text{--}63~\mathrm{eV}\) with only a weak high-energy tail; the main-edge alignment matches experiment, although the measured spectrum is broader with a stronger continuum.

\medskip

\section{Conclusion}
We establish a formally rigorous DFT framework for the lowest core-excited state, which is expected to be measured in XAS. To resolve the ambiguity inherent to core-excited DFT, we grounded the formulation in the GL theorem and introduced a penalty functional that enforces a specified core hole. Under the assumption of Eq. \eqref{eq:comm-assump} that the many-body Hamiltonian commutes with the number operator for the targeted core state, the GL theorem and the KS variational principle provide a formally exact foundation for the explicit-core $\Delta$SCF scheme, establishing the core-excited KS equation of Eq. \eqref{eq:KS_org} with a penalty term. This explicit-core treatment captures core relaxation without resorting to core-hole PPs. Because the initial and final states are treated within the same PPs and share a common energy reference, the present formulation enables shift-free absolute edge alignment across different compounds. In parallel, we reformulated the dipole matrix element using a selection-rule-aware basis (SRB), reducing the formal cost from $\mathcal{O}(N^4)$ to $\mathcal{O}(N^3)$, while providing a clear geometric interpretation. More broadly, the derivation shows how seemingly counterintuitive expansions over intermediate basis—combined with fermionic anticommutation and cancellations from prohibited double occupancy—can leverage the Pauli exclusion principle to a more compact and symmetry-transparent form. 
This computational reduction is crucial not merely for efficiency but for feasibility: it makes site-resolved XAS in large supercells and complex solids practical.
Benchmarks for \ce{C}, \ce{B}, \ce{O}, and \ce{Li} K-edges reproduce the main features of the line shapes, polarization anisotropy, and absolute energy alignment (with \ce{BN} as a notable exception), demonstrating that our calculations provide a robust, site-resolved spectroscopic fingerprint that captures the phase-dependent trends across chemically distinct sites. Future extensions include explicit $\mathbf{k}$-point sampling beyond $\Gamma$-only supercells, L/M-edge spectroscopy in non-collinear systems, and the treatment of multiple core-excited final states.

\subsection{ACKNOWLEDGMENTS}
The authors gratefully acknowledge the members of the Ozaki laboratory for constructive comments and H.~Kawai (Niigata University) for valuable discussions about pseudopotentials.
S.~A.\ acknowledges financial support from the Research Assistant (RA) program of the Institute for Solid State Physics (ISSP), The University of Tokyo.
Numerical calculations were performed in part using the supercomputer facilities of the Institute for Solid State Physics (ISSP), The University of Tokyo.

\bibliography{XAS-refs}

\end{document}